\documentclass[10pt,letter,final]{IEEEtran}
\usepackage{amsthm,amssymb,graphicx,multirow,amsmath,color,amsfonts}
\usepackage[update,prepend]{epstopdf}
\usepackage[noadjust]{cite}
\usepackage{tikz}
\usetikzlibrary{arrows,calc}		
\usepackage{bbm} 
\usepackage{pdfpages}
\usepackage{tabulary}
\usepackage{multirow}
\usepackage{comment}
\usepackage{dsfont}
\usepackage{pgf,tikz}
\usepackage{paralist}
\usepackage{amsfonts}
\usepackage{setspace,lipsum}
\usepackage{url}
\usepackage{soul}
\usepackage{bigints}
\usepackage{mathtools}
\usepackage[english]{babel}
\usepackage[autostyle]{csquotes}
\usepackage{textcomp}
\usepackage{varwidth}
\usepackage{tcolorbox}
\usepackage{lipsum}
\usepackage[normalem]{ulem}

\usepackage{cite}

\ifCLASSINFOpdf
\else
\fi
\begin{document}
\newtheorem{theorem}{Theorem}
\newtheorem{corollary}{Corollary}
\newtheorem{lemma}{Lemma}
\newtheorem{remark}{Remark}
\title{Adaptive Rate NOMA for  Cellular IoT Networks}
\allowdisplaybreaks\vspace{-4mm}
\author{G. Sreya,~\IEEEmembership{}
        S. Saigadha,~\IEEEmembership{}
        ~Praful D. Mankar,~\IEEEmembership{}
        Goutam Das, Harpreet S. Dhillon
\thanks{G. Sreya, S. Saigadha and G. Das are with IIT Kharagpur, India (Email: \{sreya.gopakumar, saigadha000\}@gmail.com and gdas@gssst.iitkgp.ac.in).
P. D. Mankar is with SPCRC, IIIT Hyderabad, India (Email:  praful.mankar@iiit.ac.in).  H. S. Dhillon is with Wireless@VT, Department of ECE, Virginia Tech, Blacksburg, VA 24061, USA (Email: hdhillon@vt.edu). The work of H. S. Dhillon was supported by the U.S. National Science Foundation under Grant CPS-1739642. 
}}
\maketitle\vspace{-40mm}
\begin{abstract}
Internet-of-Things (IoT) technology is envisioned to enable a variety of real-time applications by interconnecting billions of sensors/devices. 
These IoT devices rely on low-power wide-area wireless connectivity for transmitting, mostly fixed- but small-size, status updates of the random processes observed by them.  
Owing to their ubiquity, cellular networks are seen as a natural candidate for providing reliable wireless connectivity to IoT devices.
Given the massive number of IoT devices, enabling non-orthogonal multiple access (NOMA) for the mobile users and IoT devices is appealing in terms of the efficient utilization of spectrum compared to the orthogonal multiple access (OMA).
For instance, the uplink NOMA can also be configured such that the mobile users adapt their transmission rates depending upon the channel conditions while the IoT devices transmit at a fixed rate.
For this setting, we analyze the ergodic capacity of the mobile users and the mean local delay of IoT devices using stochastic geometry. Our analysis demonstrates that the aforementioned NOMA configuration provides better ergodic capacity for mobile users compared to OMA when delay constraint of IoT devices is strict. We also show that NOMA supports a larger packet size at IoT devices than OMA under the same delay constraint. 
\end{abstract}
\vspace{-3mm}
\IEEEpeerreviewmaketitle
\section{Introduction}\vspace{-2mm}
The IoT networks provide a digital fabric interconnecting billions of wireless devices for exchanging application-specific information without any human intervention. 
Many IoT applications, such as smart cities and traffic surveillance, rely on the real-time processing of information received from a massive number of sensors/devices deployed over a large area.  The key research challenges for realizing such IoT applications are to facilitate flexible deployment, wide-area coverage, low power devices, and low device complexity. The cellular networks are seen as a natural candidate for providing wide coverage to IoT devices on a massive scale \cite{Dhillon16IoT}.    However, the low-cost IoT devices may not be capable of performing complex signal processing needed for the advanced antenna array communication techniques (such as millimeter communication). Besides, the  IoT devices may experience much higher pathloss if they are deployed in places like tunnels or basements or are simply located far away from the BSs. Thus, efficient link budget planning is also crucial for low-power IoT devices. For these reasons, the sub-6 GHz band is primarily being considered to support low power wide area (LPWA) links of the low-cost IoT devices \cite{ratasuk2016nb}. However, the sub-6 GHz band is  crowded with the existing mobile services. This motivates spectral resource sharing between IoT devices and mobile users \cite{Mischa_RRA}. 

Further, the  IoT devices are generally deployed to share  observations/measurements of some physical process in the form of fixed and small payloads at random intervals. As a result, the BSs require to  support small size data packet transmissions from a massive number of low-power IoT devices \cite{Dhillon14M2M, Dhillon13M2M}. 
In release 13, 3GPP LTE included enhanced machine type communications (eMTC) and narrowband IoT (NB-IoT) communication to offer narrowband LPWA links to IoT devices in the sub-6 GHz band \cite{Amitava_NBIoT,Xingqin_NBIoT}. 
On the other hand, non-orthogonal multiple access (NOMA) can be used as a viable alternative to improve spectral utilization as well as enable massive access in IoT networks \cite{Ding_NOMA}. In the  literature, the design of NOMA-based IoT networks is extensively investigated. For instance, \cite{NOMA_IoT_Shahini} presents NOMA-aided NB-IoT networks for enhanced connectivity, \cite{NOMA_IoT_Balevi} presents ALOHA-based NOMA scheme for scalable and energy-efficient deployment of IoT networks, and  \cite{NOMA_IoT_Wu} studies the performance of NOMA-based wireless powered IoT networks. However, most existing works on the design of NOMA-aided IoT networks are investigated in simplified settings, such as a single-cell system.  

Recently, stochastic geometry has emerged as a powerful tool
for modeling and analyzing a variety of large-scale wireless networks. However, works on the analysis of NOMA-aided IoT networks using stochastic geometry are relatively sparse, a few of which are briefly discussed below.  The authors of \cite{NOMA_IOT_SG_Moussa} analyze aggregators-assisted two-hop NOMA-enabled cellular IoT network by modeling the locations of IoT devices, aggregators and BSs as independent Poisson point processes
(PPPs). Therein,  aggregators are employed to relay the NOMA transmissions from the IoT devices to the BS. 
The authors of \cite{NOMA_IOT_SG_Zhou} analyze RF energy harvesting based cellular IoT networks under the PPP setting. The IoT devices first harvest energy using downlink signals and then perform the uplink data transmission using NOMA.   
While the existing works in this direction consider pairing of IoT devices for non-orthogonal access, NOMA can also offer an efficient solution to the co-existence of mobile users and IoT devices by pairing their transmissions in the same spectral resource, as considered in this paper.
The authors of \cite{NOMA_IoT_MobileUE_Mahmoud} analyze the throughput performance of NOMA-based uplink transmission of mobile users and IoT devices in cellular networks under the PPP setting. However, the authors apply random pairing (i.e., mobile user and IoT device are randomly selected for a cell), which undermines the NOMA performance gains. 

{ The authors of \cite{Ding2016} show that it is imperative to pair devices with distinctive link qualities for harnessing maximum performance gains from fixed-power NOMA.  
The authors of \cite{Zhiqiang2020} characterized the performance gain of NOMA over OMA, termed  the {\em large-scale near-far gain}, which is a result of the variation in link distances of NOMA users. 
Inspired by this,  we consider a new pairing scheme that selects a mobile user from  the {\em Johnson Mehl (JM) cells} \cite{parida2019stochastic} to ensure the mobile user with shorter link distance (i.e., good channel quality) is selected for pairing, as will be discussed shortly. In most cases, this approach will ensure distinctive link qualities of the mobile user and IoT device selected for pairing.}

\textit{Contributions:}  
This paper presents a new stochastic geometry-based analysis of uplink NOMA for the non-orthogonal transmission of mobile users and IoT devices in cellular networks with power control.  In particular, we consider adaptive rate NOMA wherein the mobile users adapt modulation and coding scheme (MCS) according to the time-varying channel and the IoT devices transmit fixed but small-size data packets. We assume that the locations of IoT devices, mobile users and BSs follow independent PPPs. Further, we consider mobile users with serving link distance below threshold $L$ for pairing to ensure the  distinct link quality criteria for harnessing the optimum NOMA performance gain \cite{Ding2016}.
As a result, the mobile user and IoT device  are selected for pairing from the   {\em Johnson Mehl (JM) cell} \cite{parida2019stochastic} and  {\em Poisson Voronoi (PV) cell}, respectively, corresponding to their associated BS.\footnote{This paper considers only a subset of the mobile users (from JM cells) for NOMA pairing, and the remaining mobile users (outside of JM cells) are assumed to be served in a conventional manner. The analysis for users outside the JM cell can be followed from \cite{haenggi2015meta} with small improvisations.}
For this setup, we first derive the moments of the {\em meta distribution} \cite{haenggi2015meta} for both mobile users and IoT devices. Next, we use these results to characterize the achievable ergodic capacity for the typical mobile user and the mean local delay observed by the typical IoT device. 
Finally, our numerical results validate the analytical findings and demonstrate  that adaptive rate NOMA  is more spectrally-efficient than OMA when the delay constraint of IoT devices is strict. \vspace{-3mm}
\section{System Model}
We assume that the locations of BSs, mobile users and IoT devices form independent homogeneous PPPs $\Phi_{\rm b}^\prime$, $\Phi_{\rm m}$ and $\Phi_{\rm t}$ of densities $\lambda_{\rm b}$, $\lambda_{\rm m}$ and $\lambda_{\rm t}$, respectively, on $\mathbb{R}^2$. 
We present the uplink analysis for the typical BS placed the origin $o$ by adding an additional point at $o$ to $\Phi_{\rm b}^\prime$. Let  $\Phi_{\rm b}=\Phi_{\rm b}^\prime\cup\{o\}$.  For more details on this typical cell viewpoint, please refer to  \cite{mankar2020typicalcell}.
Mobile users and IoT devices are assumed to associate with their nearest BSs. Thus, the mobile users and IoT devices associated with BS at $\mathbf{x}$ must lie within {\em Poisson Voronoi (PV) cell} which is $V_\mathbf{x} = \{\mathbf{y} \in \mathbb{R}^2 : \left\| \mathbf{x}-\mathbf{y} 
      \right\| \leq  \left\| \mathbf{z}-\mathbf{y} \right\|, \mathbf{z} \in \Phi_{\rm b}\}$.


It is important to pair devices with distinct link qualities to achieve NOMA benefits \cite{Ding2016}.
Therefore, we pair mobile users with serving link distances shorter than $L$ with the IoT devices. 
This ensures that the mobile users experiencing  good channel quality are involved in the NOMA pairing. 
 Thus, the  NOMA pair associated with a BS at $\mathbf{x}$ includes the mobile user within the JM cell $\mathcal{V}_\mathbf{x}=\mathcal{B}_\mathbf{x}(L)\cap V_\mathbf{x}$ \cite{parida2019stochastic} and the IoT device within  the PV cell $V_\mathbf{x}$, where $\mathcal{B}_\mathbf{x}(L)$ is a ball of radius $L$ centered at $\mathbf{x}$. Note that $L$ controls the fraction of mobile users available for pairing. This fraction is equal to $\mathcal{A}_{\rm L}=1- \exp(-\pi\lambda_{\rm b} L^2)$ \cite{mankar2020distance}, which clearly increases with $L$.

In the proposed uplink NOMA, we consider that the BS first decodes the mobile users' signal  in the presence of intra-cell interference from its paired IoT device. Next, the BS applies  successive interference cancellation (SIC) technique to remove the intra-cell interference to the IoT device from the mobile user. After that, it decodes the IoT devices' signal. { Thus, we effectively consider multi-user detection by SIC.}

{  We assume that each mobile user has perfect knowledge of its uplink signal to interference ratio (${\rm SIR}$) and can employ infinitely many MCS levels such that there is an MCS level that achieves Shannon capacity with an arbitrarily small ${\rm BER}$ for a realized ${\rm SIR}$. Under this ${\rm SIR}$ adaptive MCS selection, the transmission rate of the mobile user is $\log_2(1+\beta_{\rm m})$ when the realized ${\rm SIR}$ is $\beta_{\rm m}$.}
This is also beneficial to improve the rate of successful transmission for the IoT devices as the BS will always be able to successfully perform the SIC operation because of the mobile user's channel adaptive transmission strategy. 
We term this scheme the {\em adaptive rate NOMA}. 
The IoT devices are assumed to transmit at a fixed rate as they may not be complex enough to transmit with adaptive MCS.

This paper assumes that each BS employs NOMA transmission of IoT devices and mobile users (from JM cells) over the same spectral band and uses different spectral band for the transmission of mobile users lying outside of the JM cells.    We assume the standard power law path-loss model with exponent $\alpha$, and consider that both mobile users and IoT devices transmit using a distance-proportional fractional power control scheme. We use subscript $i\in\{{\rm m}, {\rm t}\}$ for denoting the mobile user (i.e., $i={\rm m}$) and the IoT device (i.e., $i={\rm t}$). Thus, the transmit power of device $i$ is $\rho_{i}R_i^{\alpha\epsilon_i}$ where $R_i$, $\rho_i$ and $\epsilon_i \in [0, 1]$  denote  its  serving link  distance, baseline transmit power  and power control fraction, respectively.
 Let $\Psi_{\rm t}$ and $\Psi_{\rm m}$ denote the point processes of the inter-cell interfering IoT devices and mobile users, respectively. Let $R_{\mathbf{x}_{i}}$ and $D_{\mathbf{x}_i}$ denote the distances of device $i$ located at $\mathbf{x}$ from its serving BS and the   typical BS placed at $o$. We assume independent Rayleigh fading over all links. The received ${\rm SIR}$ at the typical BS at $o$ from the mobile user in $\mathcal{V}_o$ is 
\begin{align}
&{\rm SIR_{\rm m}}= \dfrac{\rho_{\rm m} h_{\rm m}R_{\rm m}^{\alpha \left(\epsilon_{\rm m} -1\right)}}{\rho_{\rm t}h_{\rm t}R_{\rm t}^{\alpha \left(\epsilon_{\rm t} -1\right) }+I_{\rm m} +I_{\rm t}},\label{SIRm}
\end{align}
and the ${\rm SIR}$ received at the typical BS at $o$ from the IoT device in $V_o$ after removing the intra-cell interference via SIC is 
\begin{align}
&{\rm SIR_{\rm t}}=\frac{\rho_{\rm t}h_{\rm t}R_{\rm t}^{\alpha (\epsilon_{\rm t} -1)}}{I_{\rm m} +I_{\rm t}}, \text{~~where}\label{SIRt}
\end{align}   
{\small\begin{align*}
&{I_{\rm m}}= \sum\nolimits_{\mathbf{x}\in\Psi_{\rm m}} \rho_{\rm m}h_{\mathbf{x}_{\rm m}}R_{\mathbf{x}_{\rm m}}^{\alpha\epsilon_{\rm m}}D_{\mathbf{x}_{\rm m}}^{-\alpha}~
\text{and}~{I_{\rm t}} = \sum\nolimits_{\mathbf{x}\in\Psi_{\rm t}} \rho_{\rm t}h_{\mathbf{x}_{\rm t}}R_{\mathbf{x}_{\rm t}}^{\alpha\epsilon_{\rm t}}D_{\mathbf{x}_{\rm t}}^{-\alpha},
\end{align*}}
  where $h_i\sim \exp(1)$ and $h_{\mathbf x_i}\sim\exp(1)$ are the small scale fading gains of intended device and interfering device at $\mathbf{x}$, respectively, for $i\in\{{\rm m},{\rm t}\}$.

The conditional success probability (conditioned on the locations of the mobile user $\mathbf{y}_{\rm m}$, IoT device $\mathbf{y}_{\rm t}$ and the inter-cell interferers' point process $\Psi= \Psi_{\rm m}\cup \Psi_{\rm t}$) 
for the mobile user and the IoT device with ${\rm SIR}$ thresholds $\beta_{\rm m}$ and $\beta_{\rm t}$ are
\begin{align}
{\rm P_{\rm m}(\beta_{\rm m};\mathbf{y},{\rm \Psi})}&=\mathbb{P}({\rm SIR_{\rm m}}>\beta_{\rm m}|\mathbf{y},{\rm \Psi}),~\text{and}\label{CSPm}\\
{\rm P_{\rm t}(\beta_{\rm m},\beta_{\rm t};\mathbf{y},{\rm \Psi}})&=\mathbb{P}({\rm SIR_{\rm m}}>\beta_{\rm m}, {\rm SIR_{\rm t}}>\beta_{\rm t}|{\mathbf{y},{\rm \Psi}})\label{CSPt1},
\end{align}   
where $\mathbf{y} = \mathbf{y}_{\rm m}\cup\mathbf{y}_{\rm t}$. The success probability of the IoT device depends on the joint decoding of messages of both the devices. However, because of the assumption of the adaptive transmission, the mobile user's signal is always decodable at the BS with arbitrarily small error probability. Hence, its intra-cell interference to the IoT devices can be eliminated using SIC because of which \eqref{CSPt1} reduces to
\begin{equation}
{\rm P_{\rm t}(\beta_{\rm t};\mathbf{y},{\rm \Psi})}=\mathbb{P}( {\rm SIR_{\rm t}}>\beta_{\rm t}|{\mathbf{y},{\rm \Psi})}.\label{CSPt}
\end{equation}  
The distribution of conditional success probability, termed  {\em meta distribution} \cite{haenggi2015meta}, is useful in studying the network performance in terms of the percentage of devices experiencing success probability above some pre-defined threshold. Hence, we aim to derive the meta distributions for both the mobile user and IoT device under the aforementioned NOMA strategy. 

Under the adaptive transmission strategy, the {\em ergodic rate} of the typical mobile user is
\begin{align}
\mathcal{R}_{\rm m}&=\mathbb{E}[\log_2(1+{\rm SIR_m})].
\end{align}
As the IoT devices are deployed to transmit their observations in a timely manner, it is meaningful to characterize their performance using the {\em mean local  delay}. The mean local  delay is defined in \cite{haenggi2015meta} as the mean number of transmissions needed for the successful delivery of a packet.\vspace{-2mm}
\section{Analysis of Adaptive Rate NOMA}
 The link distance distribution and the point processes of the inter-cell interfering devices are crucial for the meta distribution analysis, which we will discuss next. Recall, we assume that the paired mobile  user and IoT device are located uniformly at random within $\mathcal{V}_o$ and  $V_o$, respectively. 
 The probability density function (${\rm pdf}$)  of the  link distance $R_{\rm t}$ of  IoT device can be approximated as
\begin{equation}
      f_{R_{\rm t}}(r) =  2\pi\rho\lambda_{\rm b}r \exp(-\pi\rho\lambda_{\rm b} r^2),\hspace{0.3cm} r\geq0\hspace{0.1cm}, \label{fRt}
\end{equation}
where $\rho = 9/7$ \cite{mankar2020distance}. The serving link distance $R_{\rm m}$ of the  mobile user is bounded by $L$ as it is selected from $\mathcal{V}_o$. Hence, its ${\rm pdf}$ can be obtained by truncating \eqref{fRt} as
\begin{equation}
    f_{R_{\rm m}}(r) = \frac{2\pi\rho\lambda_{\rm b}r \exp\left(-\pi\rho\lambda_{\rm b}r^2\right)}{1-\exp\left(-\pi\rho\lambda_{\rm b}L^2\right)}, ~~~0\leq r \leq L.\label{fRm}
\end{equation} 
Now, we characterize the inter-cell interferers' point processes $\Psi_{\rm m}$ and $\Psi_{\rm t}$ in the following. Both these processes are non-stationary since the inter-cell interfering devices lie outside $V_o$. It is well-known that the exact characterization of uplink interferers' point process is difficult. However, an accurate approximation of the pair correlation function (${\rm pcf}$) of $\Psi_{\rm m}$ as seen from the typical BS is derived in \cite{parida2019stochastic} as $g_{\rm m}(r)=1-\exp(-2\pi\bar{\mathcal{V}}_o^{-1}r^2)$,
where $\bar{\mathcal{V}}_o^{-1}=\mathbb{E}[|\mathcal{V}_o|^{-1}]$ and $|A|$ denotes the area of set $A$.  Using this ${\rm pcf}$ and the fact that there is a single interfering user from each cell, we can approximate $\Psi_{\rm m}$ using a non-homogeneous PPP with density
\begin{equation}
\tilde{\lambda}_{\rm m}(r)=\lambda_{\rm b}g_{\rm m}(r).\label{pcf_Mobile}
\end{equation}
The ${\rm pcf}$ of $\Psi_{\rm t}$ can be obtained simply by replacing $\bar{\mathcal{V}}_o^{-1}$  with $\mathbb{E}[|V_o|^{-1}]\approx\frac{7}{5}\lambda_{\rm b}$ (which corresponds to the case $L\to\infty$) as $g_{\rm t}(r)=1- \exp(-\frac{14}{5}\pi\lambda_{\rm b}{r}^2)$,
 which exactly matches with the ${\rm pcf}$ derived in \cite{haenggi2017user}.
Thus, similar to  $\Psi_{\rm m}$, we can also approximate  $\Psi_{\rm t}$ using a non-homogeneous PPP with density 
\begin{equation}
    \tilde{\lambda}_{\rm t}(r)=\lambda_{\rm b}g_{\rm t}(r).\label{pcf_Iot}
\end{equation}
  Now, in the following, we analyze the meta distributions  of ${\rm SIR_m}$ and ${\rm SIR_t}$. It is well-known that the exact expression for meta distribution is difficult to derive. Hence, similar to \cite{haenggi2015meta}, we focus on deriving the  moments of these meta distributions.\vspace{-1mm}
\begin{theorem}
\label{thm_Mb_m}
The b-th moment of meta-distribution of the typical mobile user under the adaptive rate NOMA  is 
\begin{equation}
     M^{\rm m}_{\rm b} = \mathbb{E}_{R_{\rm m}}\hspace{-1mm}\left[\mathcal{I}_1(s_{\rm m})\mathcal{I}_2(s_{\rm m})\mathcal{M}(s_{\rm m})\right], \label{Mbm}
\end{equation}
where{\small $s_{\rm m}=\frac{\beta_{\rm m}}{\rho_{\rm m}}R_{\rm m}^{\alpha (1-\epsilon_{\rm m})}$,
   \begin{align*}
  \mathcal{I}_1(s_{\rm m}) &= \mathbb{E}_{R_{\rm t}}\hspace{-1mm}\left[\left(1+s_{\rm m}\rho_{\rm t}R_{\rm t}^{\alpha (\epsilon_{\rm t} -1) }\right)^{-b}\right],\\
  \mathcal{I}_2(s_{\rm m}) &=  \exp\left(-2\pi\int_0^\infty\tilde{\lambda}_{\rm t}(u)\left(1-\int_0^{u} (1+\right.\right.\\
  &\qquad \left.\vphantom{\int_0^\infty}\left.\vphantom{\int_0^u} \hspace{5mm}{s_{\rm m}\rho_{\rm t}r^{\alpha \epsilon_{\rm t} }}u^{-\alpha})^{-b}
f_{R_{\mathbf x_{\rm t}}}(r\vert u){\rm d}r\right)u{\rm d}u\right),\\
  \mathcal{M}(s_{\rm m}) &=  \exp\left( -2\pi\int_0^\infty\hspace{-1mm}\tilde{\lambda}_{\rm m}(u)\hspace{2mm}\left(1-\int_0^{{\rm min}(u,L)} \hspace{-1cm}(1+\right.\right.\\
  &\qquad \left.\vphantom{\int_0^\infty}\left.\vphantom{\int_0^\infty}\hspace{5mm}{s_{\rm m}\rho_{\rm m}r^{\alpha \epsilon_{\rm m} }}u^{-\alpha})^{-b} f_{R_{\mathbf x_{\rm m}}}(r\vert u){\rm d}r\right)u{\rm d}u\right),
  \end{align*}}
  \hspace{-2mm}and  ${\rm pdf}$s of $R_{\rm t}$, $R_{\rm m}$, $R_{\mathbf{x}_{\rm t}}$ and $R_{\mathbf{x}_{\rm m}}$ are given in \eqref{fRt}, \eqref{fRm}, \eqref{pdf_Rxt_Dxt} and \eqref{pdf_Rxm_Dxm}.
\end{theorem}
\begin{proof}
 Please refer to the Appendix for the proof.
\end{proof}\vspace{-2mm}
\begin{corollary}
 The \textit{b}-th moment of meta-distribution of the typical mobile user under OMA is 
\begin{equation}
        \tilde{M}^{\rm m}_{\rm b} = \mathbb{E}_{R_{\rm m}}\left[\mathcal{M}(s_{\rm m})\right],
\end{equation}
where $s_{\rm m}$ and $\mathcal{M}(s_{\rm m})$ are given in Theorem 1.
\end{corollary}\vspace{-1mm}
Now, we present moments of meta distributions for the IoT device under the adaptive rate NOMA  and OMA strategies.\vspace{-1mm}
\begin{theorem}
\label{thm_IoT}
The b-th moment of meta-distribution of the typical IoT device under the adaptive rate NOMA  is
\begin{equation}
     M^{\rm t}_{\rm b} = \mathbb{E}_{R_{\rm t}}\hspace{-1mm}\left[\mathcal{I}_2(s_{\rm t})\mathcal{M}(s_{\rm t}) \right],\label{Mbt}
\end{equation}
where $s_{\rm t} = \frac{\beta_{\rm t}}{\rho_{\rm t}}R_{\rm t}^{\alpha (1-\epsilon_{\rm t})}$, $\mathcal{I}_2(s_{\rm t})$ and $\mathcal{M}(s_{\rm t})$ are given in \eqref{Mbm}.
 \end{theorem}
\begin{proof}
\newcommand\myeqa{\stackrel{\mathclap{\normalfont\mbox{(a)}}}{=}}
  From \eqref{CSPt}, the conditional coverage probability of the typical IoT device located at $\mathbf{y}_{\rm t}$ is
  \begin{align*}
    {\rm P_t}(\beta_{\rm t}&;\mathbf{y},{\rm \Psi})= \mathbb{P}\left(h_{\rm t} > I_{\rm m}s_{\rm t} + I_{\rm t}s_{\rm t}|\mathbf{y},\Psi\right),\\
      &\stackrel{(a)}{=}\prod\limits_{\mathbf{x}\in{\rm \Psi_m}}\frac{1}{1+s_{\rm t}\rho_{\rm m}R_{\mathbf x_{\rm m}}^{\alpha\epsilon_{\rm m}}D_{\mathbf x_{\rm m}}^{-\alpha}}\prod\limits_{\mathbf{x}\in{\rm \Psi_t}}\frac{1}{1+s_{\rm t}\rho_{\rm t}R_{\mathbf x_{\rm t}}^{\alpha\epsilon_{\rm t}}D_{\mathbf x_{\rm t}}^{-\alpha}},
  \end{align*}
  where (a) follows from the assumption that $h_{\rm t},h_{\mathbf{x}_{\rm m}}$ and $h_{\mathbf{x}_{\rm t}} \sim \exp(1)$ and since $\Psi_{\rm m}$ and $\Psi_{\rm t}$ are independent.
  
  Now, \textit{b}-th moment of meta distribution can be obtained as
$M^{\rm t}_{\rm b} = \mathbb{E}[{\rm P^b_{\rm t}(\beta_{\rm t};\mathbf{y},{\rm \Psi})}]$.
  Further, following the similar steps given in the proof of Theorem \ref{thm_Mb_m}, we obtain \eqref{Mbt}.
\end{proof}
\begin{corollary}
 The \textit{b}-th moment of meta-distribution of the typical IoT device under OMA is given by
\begin{equation}
 \tilde{M}^{\rm t}_{\rm b} = \mathbb{E}_{R_{\rm t}}\left[\mathcal{I}_2(s_{\rm t})\right],
\end{equation}
where $s_{\rm t}$ and $\mathcal{I}_2(s_{\rm t})$ are given in Theorem \ref{thm_IoT}.
\end{corollary}
 The first moment of the conditional success probability is  the spatially averaged distribution of  ${\rm SIR}$. Thus, the complementary ${\rm CDF}$s of ${\rm SIR_m}$ under NOMA and OMA becomes
 \begin{equation}
     \bar{F}_{\rm m}(\beta_{\rm m})=M_{1}^{\rm m} \text{~~and~~} \tilde{F}_{\rm m}(\beta_{\rm m})=\tilde{M}_{1}^{\rm m},\label{CCDF_SIR}
 \end{equation}
 respectively. In OMA, each BS is considered to schedule its associated mobile users  and IoT devices for $\eta$ and $1-\eta$ fractions of time.
 Using \eqref{CCDF_SIR}, we now present the ergodic rate of the typical mobile user in the following theorem.
 \begin{corollary}
 Ergodic rates of the typical mobile user under NOMA and OMA, respectively, are  
\begin{align}
\mathcal{R}_{\rm m}&=\frac{1}{\ln(2)}\int_0^\infty\frac{1}{1+\gamma}\bar{F}_{\rm m}(\gamma){\rm d}\gamma,\label{ErgRateNOMA}\\
\text{and~}\tilde{\mathcal{R}}_{\rm m}&=\frac{\eta}{\ln(2)}\int_0^\infty\frac{1}{1+\gamma}\tilde{F}_{\rm m}(\gamma){\rm d}\gamma.\label{ErgRateOMA}
\end{align}
 \end{corollary}
 \begin{corollary}
 Mean local delay of the typical IoT device under NOMA and OMA, respectively, are 
  \begin{align}
      &{\rm D_t(\beta_{\rm t})} = M^{\rm t}_{\rm -1}
     ~\text{and}~{\rm \tilde{D}_t(\beta_{\rm t})} =(1-\eta)^{-1} \tilde{M}^{\rm t}_{\rm -1}.\label{MLDOMA}
  \end{align}
 \end{corollary}
The optimal selection of power control  fractions $\epsilon_{\rm m}$ and $\epsilon_{\rm t}$ is crucial to maximize the ergodic rate for the mobile user.  However, maximizing the ergodic rate of the mobile user may negatively impact the mean local delay for the IoT device. Therefore, we consider maximizing the ergodic rate of the mobile user under the constraint of maximum mean local delay of the IoT device for NOMA and OMA cases as below 
{\begin{align}
{\mathcal P}_{\rm {NOMA}}:~~&\underset{(\epsilon_{\rm m},\epsilon_{\rm t})}{\max} \mathcal{R}_{\rm m}, ~~{\rm s.t }~~ {\rm D_t(\beta_{\rm t})}\leq \tau,\label{optRateNOMA}\\
{\mathcal P}_{\rm {OMA}}:~~&\underset{(\eta,\epsilon_{\rm m},\epsilon_{\rm t})}{\max}
\tilde{\mathcal{R}}_{\rm m}, ~~{\rm s.t. }~~{\rm \tilde{D}_t(\beta_{\rm t})} \leq \tau,\label{optRateOMA}
\end{align}}
where $\tau$ represents a predefined threshold.
{
Under the fixed-rate NOMA, the successful transmission of IoT device is conditioned on the successful decoding of the mobile device's signal. Thus, the fixed-rate NOMA will lead to an inferior mean local delay performance for the IoT device compared to the adaptive rate NOMA. 
As a result, the IoT device  requires smaller transmission power (and thus smaller  intra-cell interference to the mobile user) to ensure the mean local delay  is below threshold $\tau$  under the adaptive rate NOMA compared to the fixed-rate NOMA. Therefore,  the adaptive rate NOMA provides higher ergodic rate  compared to the   throughput achievable under the fixed rate NOMA.
}
\vspace{-3mm}
\section{Numerical Results and Discussions}\vspace{-1mm}
\begin{figure*}[ht!]
  \hspace{-2mm}\includegraphics[width=0.35\textwidth]{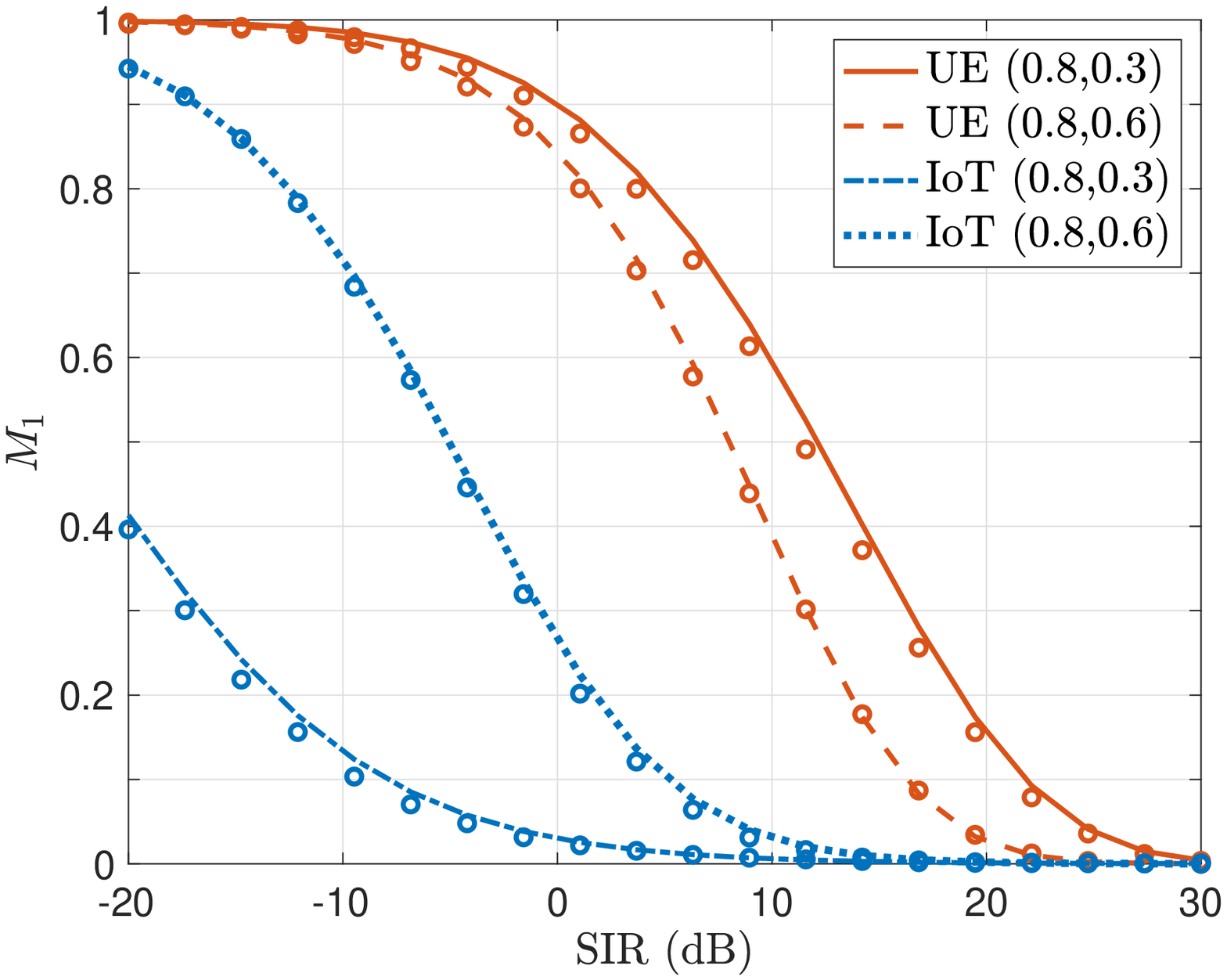}
    \hspace{-6mm}\includegraphics[width=0.35\textwidth]{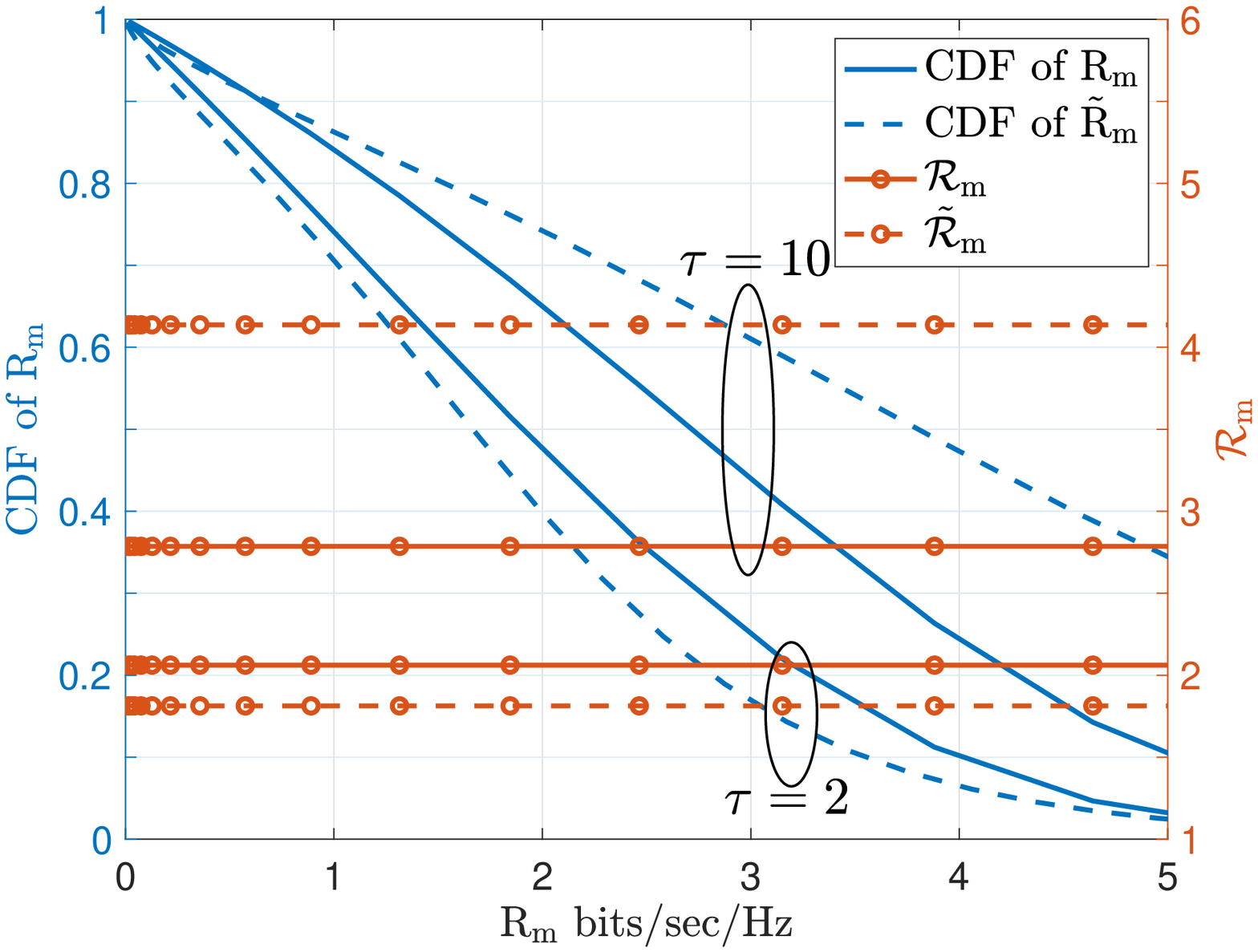}       
    \hspace{-2mm}\includegraphics[width=0.35\textwidth]{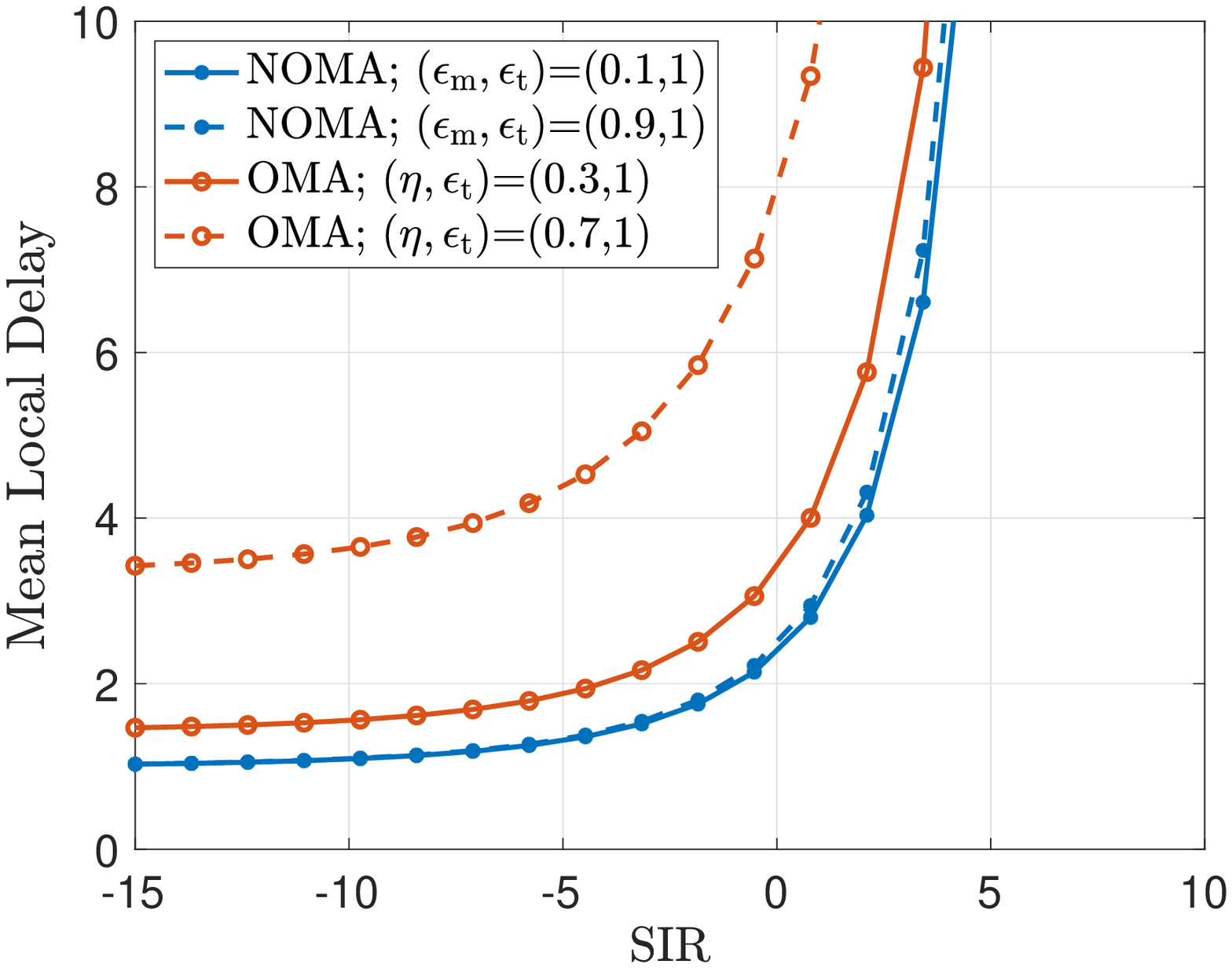}
      \vspace{-8mm}
\caption {Left: Verification of $M_1$ of both devices. The lines and markers correspond to simulation and analytical results, respectively. Middle: rate distribution and ergodic rate for optimally configured NOMA and OMA. Right: mean local delay of IoT devices for various power control fractions $(\epsilon_{\rm m}, \epsilon_{\rm t})$.}\vspace{-6mm}
\label{fig:Combined}
\end{figure*}
We consider $\lambda_{\rm b}=10^{-4}$, $\mathcal{A}_{\rm L}= 0.25$, $\alpha=4$, $\beta_{\rm t}= -5$ dB and $\rho_{\rm m} =\rho_{\rm t} = 1$, unless mentioned otherwise. 
Fig. \ref{fig:Combined} (left) verifies the accuracy of the first moment of meta distribution derived for the mobile users and  IoT devices under the adaptive rate NOMA for different values of $(\epsilon_{\rm m},\epsilon_{\rm t})$. The first moments of meta distribution of  mobile users  decreases and  IoT devices  increases with the increase in $\epsilon_{\rm t}$ for a given $\epsilon_{m}$.

{   We compare the proposed NOMA with the conventional OMA in terms of the rate distribution and optimal ergodic rate of  mobile users in Fig. \ref{fig:Combined} (middle) and the mean local delay of IoT devices in Fig. \ref{fig:Combined} (right).  Fig. \ref{fig:Combined} (middle) presents the rate distribution and ergodic rate for optimally configured NOMA and OMA. It is not surprising to see that the NOMA provides improved rate distribution compared of OMA for  $\tau=2$ (i.e., a strict delay constraint). This is because in OMA, the IoT device requires higher medium access probability  (i.e., $1-\eta$) to ensure its delay constraint when $\tau$ is small which allows smaller transmission times for mobile user. Whereas NOMA allows continuous medium access to mobile users with some interference from IoT devices. 
    Besides, the figure shows that NOMA underperforms for $\tau=10$ (i.e., a loose delay constraint).   
    This is because under OMA, the IoT device require smaller $1-\eta$ to ensure delay constraint for higher $\tau$ and thus it allows the mobile user to transmit more often.}


Fig. \ref{fig:Combined} (right) shows the mean local delay for the IoT device with full power control. It can be observed that the delay is better under NOMA compared to OMA. Besides, it is not sensitive to $\epsilon_{\rm m}$ since  SIC is always successful for the adaptive NOMA. However, the delay performance under OMA is very sensitive to $\eta$, which is expected.
The figure also shows that the mean local delay  degrades with the increase of ${\rm SIR}$ threshold $\beta{\rm_t}$ and also with the increase of $\epsilon_{\rm m}$ under NOMA and $\eta$ under OMA. It also demonstrates that for a given threshold $\tau$, NOMA can be configured such that it meets the delay constraint with a larger $\beta_{\rm t}$ compared to that under OMA case. This implies that NOMA can support a larger message size as compared to OMA under the same delay constraint. Besides, it also shows that the mean delay does not significantly change for a wide range of $\epsilon_{\rm m}$ under the NOMA whereas it drastically degrades with a moderate increase in $\eta$ under OMA.

{ Furthermore, it is expected that the ergodic rate under both NOMA and OMA degrades with the increase of $L$. This is because a larger JM cell accommodates more mobile users with lower ${\rm SIR_m}$s. The optimal fraction $\mathcal{A}_{\rm L}$ of mobile users involved in the non-orthogonal transmission with IoT devices depends on the network design parameters, such as  bandwidth partitioning for NOMA and non-NOMA users, scheduling policy, and load distributions of mobile  and IoT services. This investigation is a promising direction for future research.}
\vspace{-3mm}
\section{Conclusion}\vspace{-1mm}
 We proposed an  adaptive rate NOMA scheme for enabling massive access in cellular-supported IoT applications wherein an IoT device and a mobile user are paired for non-orthogonal transmission. The proposed {\em adaptive rate NOMA} assumes that the mobile users adapt their MCS according to the channel conditions whereas IoT devices transmit small size packets using fixed MCS. 
 Using stochastic geometry, we characterized the moments of the meta distribution for both types of devices, which are then used to characterize the ergodic rate for the typical mobile user and the mean local delay for the typical IoT device. 
Our results demonstrated that the adaptive rate NOMA provides better transmission rates for the mobile users as compared to the  OMA  under strict mean local delay constraint of IoT devices. This suggests that the proposed NOMA scheme is a spectrally-efficient solution for meeting  capacity and delay requirements of mobile users and IoT devices, respectively.


\appendix
\newcommand\myeq{\stackrel{\mathclap{\normalfont\mbox{(e)}}}{=}}


\newcommand\myeqa{\stackrel{\mathclap{\normalfont\mbox{(f)}}}{=}}

\newcommand\myeqb{\stackrel{\mathclap{\normalfont\mbox{(b)}}}{=}}
\newcommand\myeqc{\stackrel{\mathclap{\normalfont\mbox{(c)}}}{=}}
\hspace{-3mm}Letting $s_{\rm m}=\beta_{\rm m}\rho_{\rm m}^{-1}R_{\rm m}^{\alpha (1-\epsilon_{\rm m})}$, the conditional success probability of the  mobile user located at $\mathbf{y}_{\rm m}$ can be obtained as
\begin{align*}
        &{\rm P_{\rm m}(\beta_{\rm m};\mathbf{y},{\rm \Psi})}= \mathbb{P}\left(h_{\rm m} > s_{\rm m}\left(\rho_{\rm t}h_{\rm t}R_{\rm t}^{\alpha (\epsilon_{\rm t} -1) } +I_{\rm m} + I_{\rm t}\right)\hspace{0cm}|\mathbf{y},\Psi \right)\\
        &\stackrel{(a)}{=}\prod_{\mathbf{x}\in{\rm \Psi_m}}\frac{{1+s_{\rm m}\rho_{\rm t}R_{\rm t}^{\alpha (\epsilon_{\rm t} -1) }}}{1+s_{\rm m}\rho_{\rm m}R_{\mathbf x_{\rm m}}^{\alpha\epsilon_{\rm m}}D_{\mathbf x_{\rm m}}^{-\alpha}}
        \prod_{\mathbf{x}\in{\rm \Psi_t}}\frac{1}{1+s_{\rm m}\rho_{\rm t}R_{\mathbf x_{\rm t}}^{\alpha\epsilon_{\rm t}}D_{\mathbf x_{\rm t}}^{-\alpha}},
\end{align*}
where (a) follows since $h_{\rm m}, h_{\rm t},h_{\mathbf{x}_{\rm m}}$ and $h_{\mathbf{x}_{\rm t}} \sim \exp(1)$. Since $R_{\mathbf{x}_{\rm m}}\leq D_{\mathbf{x}_{\rm m}}$, for $\mathbf{x}\in\Psi_{\rm m}$, ${\rm pdf}$ of $R_{\mathbf{x}_{\rm m}}$ can be truncated as \vspace{-0.07cm}
\begin{equation}
    f_{R_{\mathbf{x}_{\rm m}}}(r|D_{\mathbf x_{\rm m}}) = \frac{2\pi\rho\lambda_{\rm b}r \exp\left(-\pi\rho\lambda_{\rm b}r^2\right)}{1-\exp(-\pi\rho\lambda_{\rm b}{\rm min}(L,D_{\mathbf x_{\rm m}})^2)}.\label{pdf_Rxm_Dxm}
\end{equation} 
Besides, $R_{\mathbf{x}_{\rm t}}\leq D_{\mathbf{x}_{\rm t}}$. Thus, the ${\rm pdf}$  of $R_{\mathbf x_{\rm t}}$ becomes 
\begin{equation}
    f_{R_{\mathbf{x}_{\rm t}}}(r|D_{\mathbf x_{\rm t}}) = \frac{2\pi\rho\lambda_{\rm b}r \exp(-\pi\rho\lambda_{\rm b}r^2)}{1-\exp(-\pi\rho\lambda_{\rm b}D_{\mathbf x_{\rm t}}^2)}, 0\leq r \leq D_{\mathbf x_{\rm t}}.\label{pdf_Rxt_Dxt}
\end{equation}
 The $b$-th moment of ${\rm P_m}(\beta_{\rm m};\mathbf{y},\Psi)$ can be obtained as\vspace{-1mm}
 \begin{align*}
     M^{\rm m}_{\rm b} &= \mathbb{E}_{R_{\rm m}}\Bigg[\mathbb{E}_{R_{\rm t}}\hspace{0cm}\left[\left(1+s_{\rm m}\rho_{\rm t}R_{\rm t}^{\alpha (\epsilon_{\rm t} -1) }\right)^{-b}\right]\vspace{-5mm}\\
     &\hspace{1cm} \mathbb{E}_{{\rm \Psi_m},R_{\mathbf{x}_{\rm m}}}\left[\underset{\mathbf{x}\in{\rm \Psi_m}}{\prod}\frac{1}{(1+s_{\rm m}\rho_{\rm m}R_{\mathbf x_{\rm m}}^{\alpha\epsilon_{\rm m}}D_{\mathbf x_{\rm m}}^{-\alpha})^b}\right]\\
     &\hspace{1cm}\mathbb{E}_{{\rm \Psi_t},R_{\mathbf{x}_{\rm t}}}\bigg[\hspace{1mm}\underset{\mathbf{x}\in{\rm \Psi_t}}{\prod}\frac{1}{(1+s_{\rm m}\rho_{\rm t}R_{\mathbf x_{\rm t}}^{\alpha\epsilon_{\rm t}}D_{\mathbf x_{\rm t}}^{-\alpha})^b}\bigg]\Bigg]\vspace{-5mm}\\
     &= \mathbb{E}_{R_{\rm m}}\Bigg[\mathbb{E}_{R_{\rm t}}\hspace{0cm}\left[\left(1+s_{\rm m}\rho_{\rm t}R_{\rm t}^{\alpha (\epsilon_{\rm t} -1) }\right)^{-b}\right]\\
     &\hspace{7mm}\mathbb{E}_{\Psi_{\rm m}}\hspace{-1mm} \left[\prod_{\mathbf{x}\in\Psi_{\rm m}}\hspace{0cm} \hspace{-1mm}\int_0^{\min(D_{\mathbf{x}_{\rm m}},L)}\hspace{-1.3cm}\vspace{0.1cm}\hspace{-2mm}\frac{1}{(1\hspace{-1mm}+\hspace{-1mm}{\rho_{\rm m}s _{\rm m}r^{\alpha \epsilon_{\rm m} }}D_{\mathbf x_{\rm m}}^{-\alpha})^b}f_{R_{\mathbf x_{\rm m}}}(r\vert D_{\mathbf x_{\rm m}}){\rm d}r\right]\\
     &\hspace{7mm}\mathbb{E}_{\Psi_{\rm t}}\hspace{-1mm} \left[\prod_{\mathbf x\in\Psi_{\rm t}}\int_0^{D_{\mathbf{x}_{\rm t}}}\hspace{-.3cm} \frac{1}{(1\hspace{-1mm}+\hspace{-1mm}{\rho_{\rm t}s _{\rm m}r^{\alpha \epsilon_{\rm t} }}D_{\mathbf x_{\rm t}}^{-\alpha})^b}f_{R_{\mathbf x_{\rm t}}}(r\vert D_{\mathbf x_{\rm t}}){\rm d}r\right]\Bigg].
     \end{align*}
      Next, using  conditional ${\rm pdf}$s of $R_{\mathbf{x}_{\rm m}}$ and $R_{\mathbf{x}_{\rm t}}$ (given in \eqref{pdf_Rxm_Dxm} and \eqref{pdf_Rxt_Dxt}), and the probability generating functional of  approximate non-homogeneous PPPs $\Psi_{\rm m}$ and $\Psi_{\rm t}$ with densities $\tilde{\lambda}_{\rm m}(r)$ and $\tilde{\lambda}_{\rm t}(r)$ (given in \eqref{pcf_Mobile} and \eqref{pcf_Iot}), we get \eqref{Mbm}.\vspace{-3mm}

\end{document}